\documentclass[conference]{IEEEtran}
\IEEEoverridecommandlockouts
\usepackage{cite}
\usepackage{amsthm, amsmath,amssymb,amsfonts,mathrsfs,bm,bbm}
\usepackage{algorithmic}
\usepackage{graphicx}
\usepackage{textcomp}
\usepackage[dvipsnames]{xcolor}
\def\BibTeX{{\rm B\kern-.05em{\sc i\kern-.025em b}\kern-.08em
    T\kern-.1667em\lower.7ex\hbox{E}\kern-.125emX}}

\newtheorem{theorem}{Theorem}
\newtheorem{proposition}{Proposition}

\newcommand{\Fc}{\mathcal{F}}

\newcommand{\Lc}{\mathcal{L}}

\newcommand{\Qc}{\mathcal{Q}}

\newcommand{\Xc}{\mathcal{X}}

\newcommand{\xv}{\vec{x}}

\newcommand{\pv}{\vec{p}}

\newcommand{\Pb}{\mathbb{P}}

\newcommand{\Eb}{\mathbb{E}}

\renewcommand{\vec}[1]{\mbox{\boldmath$#1$}}

\usepackage{xparse}
\NewDocumentCommand{\ml}{m O{M} O{Q_n^{[M]}}} {\mathcal{L}_{#1}(#2\rightarrow #3)}
\newcommand{\cmb}[2]{\begin{pmatrix}#1\\#2\end{pmatrix}}
\newcommand{\pdv}[2]{\frac{\partial #1}{\partial #2}}
\newcommand{\Ls}{\mathscr{L}}

\usepackage{tabularx}
\usepackage{diagbox}
\usepackage{lipsum}
\usepackage{multirow}
\usepackage{subcaption}
\usepackage{hhline}

\usepackage{hyperref}

\begin{document}

\title{\huge Improved Weakly Private Information Retrieval Codes\thanks{This work was supported in part by the National Science Foundation via Grants CCF-20-07067.}
}

\author{
		\IEEEauthorblockN{Chengyuan Qian, Ruida Zhou, Chao Tian, and Tie Liu}
		\IEEEauthorblockA{Department of Electrical and Computer Engineering, Texas A\&M University\\
			{\sffamily \{cyqian, ruida, chao.tian,tieliu\}@tamu.edu}
		}
	}

\textfloatsep=0.2cm
\intextsep=0.2cm
\abovecaptionskip=0.1cm
\belowcaptionskip=0.1cm

\abovedisplayskip=0.12cm
\belowdisplayskip=0.12cm 

\maketitle
\begin{abstract}
We study the problem of weakly private information retrieval (W-PIR), where a user wishes to retrieve a desired message from $N$ non-colluding servers in a way that the privacy leakage regarding the desired message's identity is less than or equal to a threshold. We propose a new code construction which significantly improves upon the best known result in the literature, based on the following critical observation. In  previous constructions, for the extreme case of minimum download, the retrieval pattern is to download the message directly from $N-1$ servers; however this causes leakage to all these $N-1$ servers, and a better retrieval pattern for this extreme case is to download the message directly from a single server. The proposed code construction allows a natural transition to such a pattern, and for both the maximal leakage metric and the mutual information leakage metric, significant improvements can be obtained. We provide explicit solutions, in contrast to a previous work by Lin et al., where only numerical solutions were obtained.
\end{abstract}

\section{Introduction}

The study of \emph{private information retrieval} (PIR) systems \cite{Chor1995} was motivated by the practical need of protecting privacy during information retrieval. 
In the canonical PIR setting, a user wishes to retrieve a message from $N$ servers, each keeping a copy of all $K$ messages. The servers are non-colluding, i.e., they cannot communicate with each other. The user wishes to ensure that the servers can infer no information about the identity of the desired message. Since the message is usually quite large, the dominant communication cost is the download from the servers.
The highest possible information bits per downloaded bit is referred to as the PIR capacity, which was recently fully characterized by Sun and Jafar \cite{Sun2017}. 
An alternative optimal code (referred to as the TCS code) was later proposed \cite{Tian2019}, which uses the minimum possible message length and query set. Many variations of the canonical PIR problem has been studied, such as colluding servers\cite{t1,t2,t3}, storage constrained \cite{c1,c2,c3,c4,c5,c6,c7,c8,c9,c10,zhu2019new}, with symmetric privacy requirement\cite{d1,d2,d3}, and with side information\cite{s1,s2,s3,s4,s5,s6,s7}. 

The perfect privacy requirement in the canonical setting can be unnecessarily stringent. A small amount of privacy leakage is likely acceptable in many practical scenarios, e.g., when the user does not mind if the server can infer the identity of the desired message with only relatively low confidence. 
This setting, where a weaker privacy constraint is placed, is referred to as weakly private information retrieval (W-PIR) \cite{Asonov2002,Toledo2016,Samy2019,ZhuqingJia2019,Lin2019,Zhou2020a,Lin2021,Samy2021,Lin2020}. 
In exchange for the loss of privacy, a higher retrieval rate can be attained, sometimes with a lower computational complexity \cite{Asonov2002}. Several different metrics have been proposed to measure the privacy leakage in W-PIR. Differential privacy was used in \cite{Toledo2016,Samy2019}, conditional entropy was used in \cite{ZhuqingJia2019}, mutual information in \cite{Lin2019}, and the maximal leakage metric (see \cite{Issa2020}) was adopted in  \cite{Zhou2020a,Lin2021}. 
The W-PIR code proposed in \cite{Zhou2020a} for the maximal leakage metric was obtained by adjusting the code proposed in \cite{Tian2019}; similar or identical code constructions were also analyzed in \cite{Lin2021} and \cite{Samy2021} under different metrics, either theoretically or numerically. 

The previously best known W-PIR code under the maximal leakage constraint \cite{Zhou2020a} was obtained by breaking the uniform distribution on the retrieval patterns in the TSC code, which increasingly favors the direct download pattern in the code as the privacy requirement is relaxed; it was shown to be optimal when $N=2$ in \cite{Lin2021}. In this work, we provide a new code construction by making the following critical observation. The direct download pattern in the TSC code essentially downloads the desired message from $N-1$ servers, one symbol from each server. However, this would result in privacy leakage to all these $N-1$ servers, when this pattern is not mixed with other patterns of retrieval. In the extreme case of minimum download, the W-PIR code in \cite{Zhou2020a} can only use this pattern, yet an alternative strategy is to directly download the full message from {\em one single server}, which only leaks to this single server. Our proposed new code utilizes this observation and allows a natural transition to this retrieval pattern.

The new code can also be viewed as adjusting the probabilities of the retrieval patterns in the TSC code, jointly with the new clean download retrieval pattern. We provide the optimal distributions explicitly under both the maximal leakage metric and the mutual information leakage metric, and show that the new code can achieve significant improvement over existing ones. It should be noted that in  \cite{Lin2021}, the simpler code without the new retrieval pattern was studied only numerically, and no explicit solution was provided.  

\section{Preliminaries}\label{sec:preliminaries}

In this section, we formally introduce the W-PIR problem under the maximum leakage metric and mutual information metric, respectively, and then review the PIR code proposed in \cite{Tian2019} that will be instrumental later on.

\subsection{Information Retrieval Systems}

There are a total of $N$ servers, and each server stores an independent copy of $K$ mutually independent messages, denoted as $W_{1:K} := (W_1, W_2, \ldots, W_K)$, where $K\geq 2$ without loss of generality. Each message consists of $L$ symbols, and each symbol is distributed uniformly in a finite set $\Xc$, which implies that
\begin{align*}
    L := H(W_1) = H(W_2) = \cdots = H(W_K),
\end{align*}
where the entropy is taken under the logarithm of base $|\Xc|$. The $i$-th symbol of the message $W_k$ is denoted as $W_k[i]$, where $i=1,\ldots,L$. An information retrieval code consists of the following component functions. When a user wishes to retrieve a message $W_k$, $k \in [1:K]$, the (random) query $Q^{[k]}_n$ sent to server-$n$ is generated according to an encoding function
\begin{align}
    Q^{[k]}_n := \phi_n(k, F), \quad \forall n \in 1:N,
\end{align}
by leveraging some private random key $F \in \Fc$. Let $\Qc_n$ be the union of all possible queries $Q^{[k]}_n$ over all $k$. For each $n \in 1:N$, upon receiving a query $q \in \Qc_n$, server-$n$ responds with an answer $A^{(q)}_n$ produced as
\begin{align}
    A^{(q)}_n := \varphi_n(q, W_{1:K}), \label{eqn:varphi}
\end{align} 
which is represented by $\ell_n^{(q)}$ symbols in certain coding alphabet $\mathcal{Y}$; to simplify the notation, we assume $\mathcal{X}=\mathcal{Y}$ in this work. We assume $\ell_n^{(q)}$ may vary according to the query but not the messages, and as such the user knows how many symbols are expected in that answer. 

For notational simplicity, we denote $A^{(Q^{[k]})}_n$ as $A^{[k]}_n$ and $\ell_n^{(Q^{[k]})}$ as $\ell_n^{[k]}$, both of which are random variables. With the answers from the servers, the user attempts to recover the message $\hat{W}_k$ using the decoding function
\begin{align}
    \hat{W}_k := \psi( A^{[k]}_{1:N}, k, F).
\end{align}
A valid information retrieval code must first satisfy $\hat{W}_k = W_k$, i.e., the desired message should be correctly recovered.

We measure the download cost by the normalized (worst-case) \emph{average download cost},
\begin{align}
    D := \max_{k \in 1:K} \Eb\left[\frac{1}{L} \sum_{n = 1}^N \ell_n^{[k]} \right],
\end{align}
where $\ell_n^{[k]}$ is the length of the answer in the code and the expectation is taken with respect  to the random key $F$.

\subsection{Maximal Leakage and Mutual Information Leakage}
The index of the desired message, denoted as $M$, is viewed as a random variable following a certain distribution. The identity of the desired message $W_M$ may be leaked to server-$n$ due to the query $Q^{[M]}_n$ sent by the user. We focus on two metrics to study this leakage. 

\vspace{0.1cm}
\noindent{\it The maximal leakage metric $\Lc(M \rightarrow Q^{[M]}_n)$: } It was shown in \cite{Issa2020} and \cite{Zhou2020a} that
\begin{align}
	\Lc(M\rightarrow Q_n^{[M]}) = \log_2\bigg{(}\sum_{q \in \Qc_n} \max_{k \in 1:K} \Pb(Q^{[k]}_n = q )\bigg{)}, \label{eq:max-leakage-def}
\end{align}
which in fact does not depend on the probability distribution of $M$. When $\Lc(M\rightarrow Q_n^{[M]})$ is large, $Q^{[M]}_n$ leaks more information of $M$ in the sense that server-$n$ can estimate $M$ more accurately; on the other hand, when $\Lc(M\rightarrow Q_n^{[M]})=0$, the retrieval is private in the sense that the distribution of $Q^{[k]}_n$ and $Q^{[k']}_n$ are identical for any $k,k'\in [1:K]$.  

A \emph{valid} code for W-PIR with $K$ messages and $N$ servers under a maximum leakage constraint $\rho$ is a collection of functions $(\{\phi_n\}_{n \in [1:N]},\{\varphi_n\}_{n \in [1:N]},\psi)$ that can correctly retrieve the desired message, and additionally satisfies the privacy constraints that for each server-$n$, 
    \begin{align}
        \Lc(M \rightarrow Q^{[M]}_n) \leq \rho,\quad n\in [1:N].\label{eq:rhodef}
    \end{align}
A download cost $D$ is called achievable for $\rho$, if there exists a valid code such that its download cost is less than or equal to $D$ for privacy constraint $\rho$. The closure of the collection of such $(\rho,D)$ pairs is called the achievable $(\rho,D)$ region under the maximal leakage constraint, denoted by $\mathcal{G}_{\text{MaxL}}$; the infimum of such achievable download cost for $\rho$ is the download-leakage function, denoted as $D_{\text{MaxL}}(\rho)$. 

\vspace{0.1cm}
\noindent{\it The mutual information metric:} In this setting, the identity of the request of message $M$ is assumed to be uniformly distributed in the set $[1:K]$. Then the mutual information leakage is 
\begin{align}
	\text{MI}(M\rightarrow Q_n^{[M]}) := I(M;Q_n), \label{eq:MI}
\end{align}
where $Q_n$ is the random query induced jointly by the random key $F$ and the random message index $M$. 
We can similarly define valid codes under the mutual information leakage constraint under the condition, 
    \begin{align}
        \text{MI}(M \rightarrow Q^{[M]}_n) \leq \rho,\quad n\in [1:N].\label{eq:rhoMIdef}
    \end{align}
Similarly the achievable $(\rho,D)$ region under this metric is denoted as $\mathcal{G}_{\text{MI}}$ and the download-leakage function as $D_{\text{MI}}(\rho)$. 

\subsection{The TSC Code}
The TSC code given in \cite{Tian2019} will serve an instrumental role  in this work. In this code, the message length $L = N-1$. A dummy symbol $W_k[0]=0$ is prepended at the beginning of all messages. In order to better facilitate the new code construction, we give a variation of the original construction, which can be viewed as probabilistic sharing among the cyclic permutations (over the $N$ servers)  of the PIR code in \cite{Tian2019}. 

Let the random key $F^*$ be a length-$K$ vector 
\begin{align}
    F^*: = (F^*_1,F^*_2,\dots,F^*_{K-1}, U),
\end{align} 
where $F^*_1,\dots,F^*_{K-1}, U$ are independent random variables uniformly distributed over the set $[0:N-1]$, i.e.,
\begin{align}
    F^* \in [0:N-1]^K \triangleq \Fc^*.
\end{align}
The query $Q_n^{[k]}$ to server-$n$ is generated by the function $\phi_n^*(k, F^*)$ defined as,
\begin{equation}
\begin{aligned}
     \phi_n^*(k, F^*) \triangleq (F^*_1, F^*_2, &\dots,F^*_{k-1}, (U+n)_N, \\
     & F^*_{k}, F^*_{k+1}, \dots, F^*_{K-1}), \label{eqn:tsc-phi}
\end{aligned}
\end{equation}
where $(\cdot)_N$ represents the modulo $N$ operation. Note that in \cite{Tian2019}, the indices $n$ and $k$ start from $0$, and the random variable $U$ here is instead $-\sum_{j = 1}^{K-1}F_j$. The new random key component $U$ is introduced here, such that the query to the server $n$ is cyclically permuted uniformly at random. As a result, for any server-$n$, the query $Q_n^{[k]}$ is uniformly distributed on the set $\Qc^* \triangleq [0:N-1]^K$, i.e.,
\begin{align}
    \Pb(Q_n^{[k]} = q) = N^{-K}, \quad \forall n\in[1:N], q\in\Qc^*. \label{eq:tscq}
\end{align}
Upon receiving this query, the server-$n$ returns the answer $A_n^{[k]}$ generated by the function $\varphi^*(q,W_{1:K})$,
\begin{align}
    \varphi^*(q,W_{1:K})& \triangleq W_1[Q_{n,1}^{[k]}]\oplus W_2[Q_{n,2}^{[k]}]\oplus \cdots \oplus W_K[Q_{n,K}^{[k]}]\notag\\
    & = W_k[(U+n)_N] \oplus \mathscr{I},
\end{align}
where $\oplus$ denotes addition in the given finite field, $Q_{n,m}^{[k]}$ represents the $m$-th symbol of $Q_{n}^{[k]}$, and $\mathscr{I}$ is the interference signal defined as
\begin{equation}
    \begin{split}
        \mathscr{I} = W_1[F_1]\oplus\cdots&\oplus W_{k-1}[F_{k-1}] \\
        &\oplus W_{k+1}[F_{k}] \oplus \cdots \oplus W_{K}[F_{K-1}].
    \end{split}
\end{equation}
Since there exists an $n \in [1:N]$, $(U+n)_N = 0$ such that $A_n^{[k]} = \mathscr{I}$, the user can retrieve the desired message $W_k$ by subtracting $\mathscr{I}$ from $A_{n'}^{[k]}$ for all $n' \not= n$. Note that with probability $N^{-(K-1)}$ the interference signal $\mathscr{I}$ consists of only dummy symbols and need not to be downloaded, in which case a direct download will be performed by retrieving the desired message from $N-1$ servers, one symbol per server. The download cost is therefore 
\begin{align}
    D^* = \frac{N}{N-1}\left(1-\frac{1}{N^{K-1}}\right) + \frac{1}{N^{K-1}} =\frac{ 1 - N^{-K} }{ 1 - N^{-1} },
\end{align}
matching the capacity result given in \cite{Sun2017}. An example of the code (with adjusted probabilities for W-PIR) is given in Section \ref{sec:forward}; more details can be found in \cite{Tian2019}.

\section{Main Results} \label{sec:main}

We summarize the performance of the new code in the following two theorems.

\begin{theorem}\label{thm:achievableML} For W-PIR under the maximal leakage constraint,  
\begin{align}
    &D_{\text{MaxL}}(\rho) \notag\\
    &\leq  1 + 
     \left( 1 - N \frac{2^{\rho} - 1}{K - 1} \right)_+ \left(\frac{1}{N} + \cdots + \frac{1}{N^{K-1}} \right),\label{eqn:def-D*}
\end{align}
where $(x)_+:=\max(x,0)$. 
\end{theorem}

The code construction for Theorem \ref{thm:achievableML} is given in the next section, which is obtained by probabilistic sharing of the TSC code with the new retrieval pattern. This result is presented in terms of the download-leakage function for the maximal leakage setting. 
For the mutual information metric setting, a few additional quantities are required to parametrize the solution.  First define a sequence $\vec{x}=(x_1,x_,\ldots,x_{K-1})$, which can be shown to be greater than or equal to 1 component-wise, using the following recursion backwards from $K-1,K-2,\ldots,1$: 
\begin{align}
&\log\frac{(K-i)x_{K-i}+i}{K} = \sum_{j=0}^{i-1}(1-N)^j \log\frac{(K-1)x_{K-1}+1}{K} \notag\\
&\,\,\quad\qquad\qquad\qquad\qquad-\sum_{j=1}^{i-1}(1-N)^j \log x_{K-i+j}.\label{eq:recursive2}
\end{align}
Since the RHS only depends on $x_{K-1},x_{K-2},\ldots,x_{K-i+1}$, the sequence is well defined, when $x_{K-1}\in[1,\infty]$ is specified. With this sequence defined, we further define the following probability vector $\pv = (p_0,p_1,\dots,p_{K-1})$:
  \begin{align}
    &p_0(\xv) = \bigg( N + N\sum_{w=1}^{K-1}\cmb{K-1}{w}(N-1)^w\prod_{j=1}^{w}\frac{1}{x_{j}}\bigg)^{-1},\label{eq:pvdef1}\\
    &p_w(\xv) = p_0(\xv) \prod_{j=1}^{w}\frac{1}{x_{j}}, \quad w\in1:K-1.\label{eq:pvdef2}
\end{align}
It can be verified that $\pv$ induces a probability distribution with the appropriate combinatorial coefficients taken into account. 

Define the region $\hat{\mathcal{G}}_{MI}$ to be the nonnegative $(\rho,D)$ pairs satisfying the following conditions
\begin{align}
    &\rho\geq \frac{1}{K}\sum_{w=0}^{K}\cmb{K}{w}(N-1)^w\notag \\
    &\,\,\Big\{ wp_{w-1} \log p_{w-1} + (K-w)p_w \log p_w\label{eq:rhopara} \\
    &\,\, - [wp_{w-1} + (K-w)p_w] \log \frac{wp_{w-1} + (K-w)p_w}{K}\Big\}\notag\\
    & D\geq \frac{N}{N-1} (1 - p_0),\label{eq:pdpara}
\end{align}
for some $x_{K-1}\in[1,\infty]$; for simplicity, $p_{-1}$ and $p_K$ are defined as zero. We then have the following theorem. 

\begin{theorem}\label{thm:achievableMI} For the mutual information leakage metric, $\text{conv}(\hat{\mathcal{G}}_{\text{MI}}\cup \{(\frac{\log K}{N},1)\})\subseteq \mathcal{G}_{\text{MI}}$, where $\text{conv}(\cdot)$ is the convex hull operation.
\end{theorem}

\begin{figure}[tb]
  \begin{subfigure}[b]{0.475\columnwidth}
    \includegraphics[scale = 0.3]{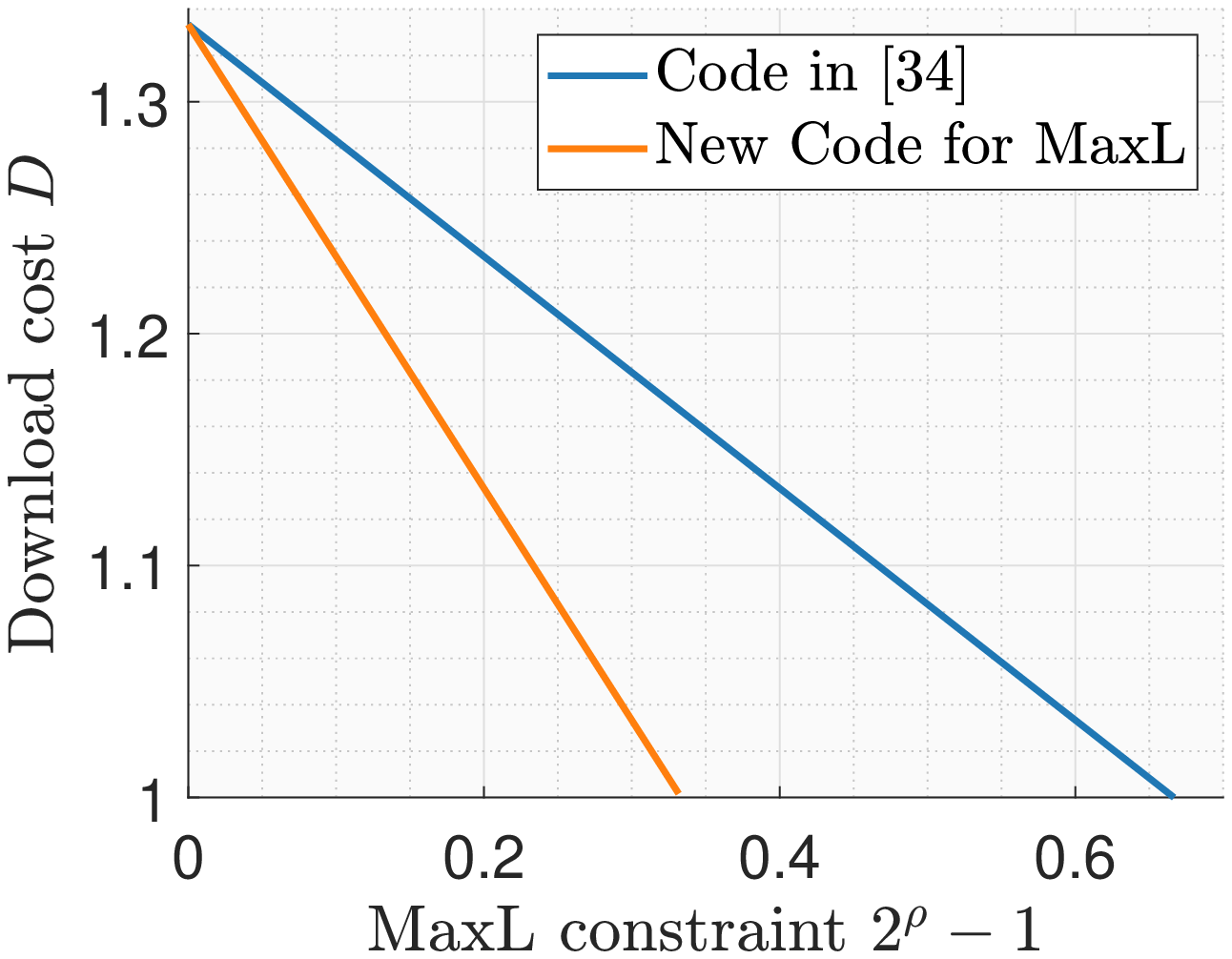}
    \caption{Maximal leakage metric.}
    \label{fig:dmi}
  \end{subfigure}
  \begin{subfigure}[b]{0.475\columnwidth}
    \includegraphics[scale = 0.3]{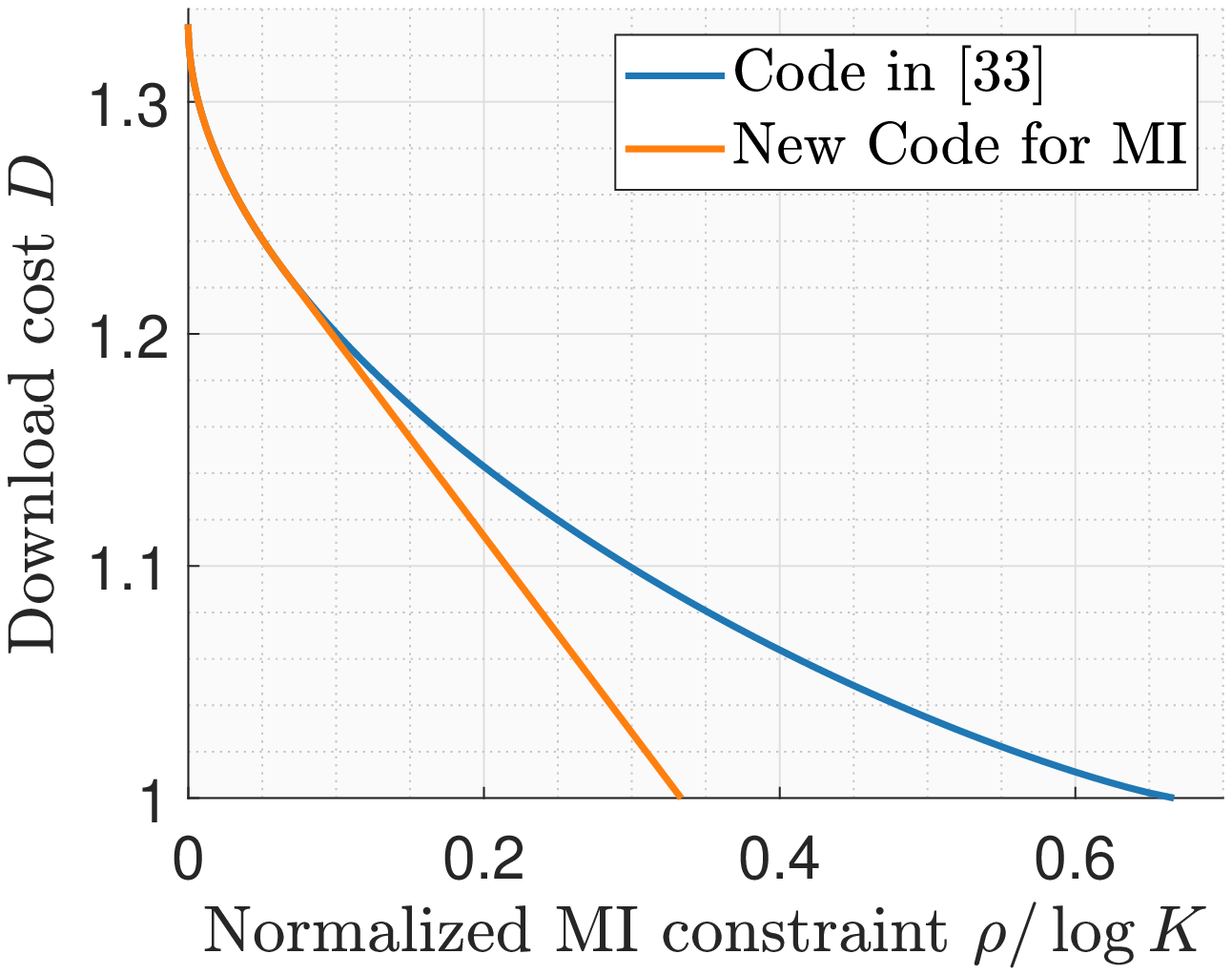}
    \caption{Mutual information metric.}
    \label{fig:dmi}
  \end{subfigure}
  \caption{Numerical comparisons between the proposed code and previously best known code under different privacy metrics when $N = 3,K = 2$. \label{fig:rhoD}}
\end{figure}

The performance of the proposed code is illustrated in Fig. \ref{fig:rhoD}. In both cases, the new extreme point of minimum download provides a new anchor point for the $(\rho,D)$ tradeoff. In essence, the new achievable regions can be obtained by proper probabilistic sharing of the existing code with the clean retrieval pattern for this new extreme point (see details in the next section), in the correct coordinate for the two metrics. The sharing structure is more sophisticated for the mutual information leakage case, and it can be seen that when $\rho$ is below a threshold, the new clean download pattern is in fact not effective, meaning it is not utilized during retrieval.

\section{A New Code Construction}\label{sec:code}
\label{sec:forward}

We first give an example to illustrate the proposed code based on probabilistic sharing, then present the general construction. 

\subsection{An Illustrative Example}

Consider the case with $K=2$ messages and $N=3$ servers. The message length is $L=N-1=2$, and we write $W_1 = (a_1,a_2)$, $W_2 = (b_1,b_2)$. 
The queries and answers are given in Table \ref{tab:dd}. The queries in the top three rows of the two tables directly request the full message from a single server denoted by $\#_1$ and $\#_2$, while the remaining nine rows are essentially the TSC code with different probabilities for the queries, assigned according to their interference signals. Note that the interference signal is controlled by the first-$(K-1)$ entries $F^*_{1:K-1}$ of the random key $F^*$. For any $F \in \Fc^*$, denote $|F|$ as the \emph{size} of the interference corresponding to random key $F$, which is also the hamming weight of $F^*_{1:K-1}$. 
In this example $|F|$ can only be $0$ or $1$.

\setlength\tabcolsep{3pt}
\begin{table}[tb]
    \caption{Proposed code for $N=3,K=2$}
    \label{tab:dd}
    \begin{subtable}[h]{0.5\textwidth}
        \centering
        \caption{Retrieval of $W_1$}
        \label{tab:dd1}
        \begin{tabular}{ |c|c|c|c|c|c|c|c|c| }
        \hline
        \rule{0pt}{2.1ex}\multirow{2}{*}{\rule{0pt}{2.5ex}{$\Pb_F(F)$}} & \multirow{2}{*}{\rule{0pt}{2.5ex}{$F$} } & \multicolumn{2}{|c|}{\textbf{Server $1$}} & \multicolumn{2}{|c|}{\textbf{Server $2$}} & \multicolumn{2}{|c|}{\textbf{Server $3$}}\\
        \cline{3-8} 
        \rule{0pt}{2.8ex}& & $Q^{[1]}_1$ & $A_1$ & $Q^{[1]}_2$ & $A_2$ & $Q^{[1]}_3$ & $A_3$ \\[0.25ex]
        \hhline{|=|=|=|=|=|=|=|=|}
        \rule{0pt}{2.25ex}$p'_0$ & ${1}$ & \#$_1$ & $a_1,a_2 $ & $\underline{00}$ & $\varnothing$ & $\underline{00}$ & $\varnothing$ \\ [0.25ex] 
        \hline
        \rule{0pt}{2.25ex}$p'_0$ & ${2}$ & $\underline{00}$ & $\varnothing$ & \#$_1$ & $a_1,a_2 $ & $\underline{00}$ & $\varnothing$ \\ [0.25ex] 
        \hline
        \rule{0pt}{2.25ex}$p'_0$ & ${3}$ & $\underline{00}$ & $\varnothing$ & $\underline{00}$ & $\varnothing$ & \#$_1$ & $a_1,a_2 $ \\ [0.25ex] 
        \hline\hline
        \rule{0pt}{2.25ex}$p_0$ & 00 & $10$ & $a_1$ & $20$ & $a_2 $ & $\underline{00}$ & $\varnothing$ \\ [0.25ex] 
        \hline
        \rule{0pt}{2.25ex}$p_0$ & 10 & $20$ & $a_2 $ & $\underline{00}$ & $\varnothing$ & $10$ & $a_1$ \\ [0.25ex]  
        \hline
        \rule{0pt}{2.25ex}$p_0$ & 20 & $\underline{00}$ & $\varnothing$ & $10$ & $a_1$ & $20$ & $a_2 $ \\ [0.25ex]  
        \hline
        \rule{0pt}{2.25ex}$p_1$ & 01 & $11$ & $a_1 \oplus b_1$ & $21$ & $a_2 \oplus b_1$ & $01$ & $b_1$ \\ [0.25ex] 
        \hline
        \rule{0pt}{2.25ex}$p_1$ & 11 & $21$ & $a_2 \oplus b_1$ & $01$ & $b_1$ & $11$ & $a_1 \oplus b_1$ \\ [0.25ex] 
        \hline
        \rule{0pt}{2.25ex}$p_1$ & 21 & $01$ & $b_1$ & $11$ & $a_1 \oplus b_1$ & $21$ & $a_2 \oplus b_1$ \\ [0.25ex] 
        \hline
        \rule{0pt}{2.25ex}$p_1$ & 02 & $12$ & $a_1 \oplus b_2$ & $22$ & $a_2 \oplus b_2$ & $02$ & $b_2$ \\ [0.25ex] 
        \hline
        \rule{0pt}{2.25ex}$p_1$ & 12 & $22$ & $a_2 \oplus b_2$ & $02$ & $b_2$ & $12$ & $a_1 \oplus b_2$ \\ [0.25ex] 
        \hline
        \rule{0pt}{2.25ex}$p_1$ & 22 & $02$ & $b_2$ & $12$ & $a_1 \oplus b_2$ & $22$ & $a_2 \oplus b_2$ \\ [0.25ex] 
        \hline
        \end{tabular}
    \end{subtable}
    \newline
    \vspace*{1em}
    \newline
    \begin{subtable}[h]{0.5\textwidth}
        \centering
        \caption{Retrieval of $W_2$}
        \label{tab:dd2}
        \begin{tabular}{ |c|c|c|c|c|c|c|c|c| }
        \hline
        \rule{0pt}{2.1ex}\multirow{2}{*}{\rule{0pt}{2.5ex}{$\Pb_F(F)$}} & \multirow{2}{*}{\rule{0pt}{2.5ex}{$F$} } & \multicolumn{2}{|c|}{\textbf{Server $1$}} & \multicolumn{2}{|c|}{\textbf{Server $2$}} & \multicolumn{2}{|c|}{\textbf{Server $3$}}\\
        \cline{3-8} 
        \rule{0pt}{2.8ex}& & $Q^{[2]}_1$ & $A_1$ & $Q^{[2]}_2$ & $A_2$ & $Q^{[2]}_3$ & $A_3$ \\[0.25ex]
        \hhline{|=|=|=|=|=|=|=|=|}
        \rule{0pt}{2.25ex}$p'_0$ & $ {1}$ & $\text{\#}_2$ & $b_1,b_2 $ & $\underline{00}$ & $\varnothing$ & $\underline{00}$ & $\varnothing$ \\ [0.25ex] 
        \hline
        \rule{0pt}{2.25ex}$p'_0$ & $ {2}$ & $\underline{00}$ & $\varnothing$ & $\text{\#}_2$ & $b_1,b_2 $ & $\underline{00}$ & $\varnothing$ \\ [0.25ex] 
        \hline
        \rule{0pt}{2.25ex}$p'_0$ & $ {3}$ & $\underline{00}$ & $\varnothing$ & $\underline{00}$ & $\varnothing$ & $\text{\#}_2$ & $b_1,b_2 $ \\ [0.25ex] 
        \hline\hline
        \rule{0pt}{2.25ex}$p_0$ & 00 & $01$ & $b_1$ & $02$ & $b_2 $ & $\underline{00}$ & $\varnothing$ \\ [0.25ex]  
        \hline
        \rule{0pt}{2.25ex}$p_0$ & 01 & $02$ & $b_2 $ & $\underline{00}$ & $\varnothing$ & $01$ & $b_1$ \\ [0.25ex]  
        \hline
        \rule{0pt}{2.25ex}$p_0$ & 02 & $\underline{00}$ & $\varnothing$ & $01$ & $b_1$ & $02$ & $b_2 $ \\ [0.25ex]  
        \hline
        \rule{0pt}{2.25ex}$p_1$ & 10 & $11$ & $a_1 \oplus b_1 $  & $12$ & $a_1 \oplus b_2 $ & $10$ & $a_1$\\ [0.25ex] 
        \hline
        \rule{0pt}{2.25ex}$p_1$ & 11 & $12$ & $a_1 \oplus b_2 $ & $10$ & $a_1$ & $11$ & $a_1 \oplus b_1 $ \\ [0.25ex] 
        \hline
        \rule{0pt}{2.25ex}$p_1$ & 12 & $10$ & $a_1$ & $11$ & $a_1 \oplus b_1 $ & $12$ & $a_1 \oplus b_2 $ \\ [0.25ex] 
        \hline
        \rule{0pt}{2.25ex}$p_1$ & 20 & $21$ & $a_2 \oplus b_1$ & $22$ & $a_2 \oplus b_2$ & $20$ & $a_2$ \\ [0.25ex] 
        \hline
        \rule{0pt}{2.25ex}$p_1$ & 21 & $22$ & $a_2 \oplus b_2$ & $20$ & $a_2$ & $21$ & $a_2 \oplus b_1$ \\ [0.25ex] 
        \hline
        \rule{0pt}{2.25ex}$p_1$ & 22 & $20$ & $a_2$ & $21$ & $a_2 \oplus b_1$ & $22$ & $a_2 \oplus b_2$ \\ [0.25ex] 
        \hline
        \end{tabular}
    \end{subtable}
\end{table}

We have omitted the dummy symbols  $a_0$ and $b_0$ for conciseness. 
The random key $F$ has a total of $12$ possible realizations, with the probability parametrized by $(p'_0, p_0, p_1)$, where $p'_0$ is the probability of direct download from a single given server, $p_0$ is that of the interference having hamming weight $0$, $p_1$ that of the interference having haming weight $1$.

\subsection{General Code Construction}
For general W-PIR with parameter $(N,K)$, we set $L=N-1$. The random key $F$ is generated from set $\Fc$ with a probability distribution $\Pb_F(F)$, where $\Fc = \Fc^* \cup [1:N] = [0:N-1]^{K} \cup  [1:N]$, and
\begin{align}
    \Pb_F(F) = \begin{cases} p'_0, & \forall F \in [1:N]\\
    p_w, & \forall F \in \Fc^*, |F| = w,~ w \in [0:K-1]
    \end{cases},
\end{align}
which needs to satisfy
\begin{align}
    N p'_0 + N\sum_{w=0}^{K-1} \binom{K-1}{w}(N-1)^w  p_w = 1.
\end{align}
The query $Q_n^{[k]}$ to server-$n$ is produced as:
\begin{align}
    Q_n^{[k]} = \begin{cases} \text{\#}_k, & F = n\\ 
    \underline{0_K}, & \forall F \in [1:N],~ F \neq n\\ 
    \phi_n^*(k,F), & \forall F \in \Fc^*\end{cases},
\end{align}
where $\underline{0_K}$ is the length-$K$ all-zero vector. The answer $A_n^{[k]}$ from server-$n$ is generated as
\begin{align}
    A_n^{[k]} = \begin{cases} W_k, & q = \text{\#}_k\\ 
    \varphi^*(q,W_{1:K}), & \forall q \in \Qc^*\end{cases}.
\end{align}
The correctness of the code is obvious, and the download cost $D$ can be simply computed as 
\begin{align}
    p_d &\triangleq N (p'_0 + p_0)\label{eq:pd}\\
    D & = p_d + \frac{N}{N-1}(1-p_d),
\end{align}
where $p_d$ is the overall probability of using a direct download. We defer the analysis of the privacy for the two metrics to the next subsection.

\section{Code Optimization and Performance Analysis}

We have provided the new code construction in a general form in the previous section, however, without optimizing the probability distribution. In this section, we optimize 
the probability distributions for the two leakage constraints, respectively.

\subsection{Optimizing for Maximal Leakage}

Since $p_d$ is directly related to $D$ in the proposed code, setting $p_d$ is equivalent to specifying a target download cost $D$ in this code. Therefore, the constrained minimization problem can be written as follows: 
\begin{align}
    \text{Minimize:}~~ & \Lc(M\rightarrow Q_n^{[M]}),\\
    \text{Variables:}~~ & \pv = (p_0,p_1,\dots,p_{K-1}),\\
    \text{Subject to:}~~ & -p_w \leq 0, \forall w\in 0:K-1,\\
    & N p_0 - p_d\leq 0,\\
    & \begin{aligned}
    N\sum_{w=0}^{K-1}\binom{K-1}{w}(&N-1)^w p_w\\
    &+ (p_d - N p_0) -1 = 0.\label{eq:probsum}
    \end{aligned}
\end{align}
This optimization problem can indeed be solved (see a proof of a similar nature in \cite{Zhou2020a}), for which the solution is 
\begin{align}
    p'_0 &= \min\left(\frac{1}{N},\frac{2^\rho-1}{K-1}\right),\\
    p_w &= \frac{1-Np'_0}{N^K},\quad w \in 0:K-1,
\end{align}
and the leakage can thus be found as
\begin{align}
    & \ml{n} = \log_2\sum_{q_n\in\Qc_n}\max_{j\in[1:K]}\Pb_{Q^{[j]}_n}(q_n)\notag\\
    & = \log_2 \left[ \left(\sum_{q_n = \underline{0_K} } + \sum_{\substack{q_n \in \Qc^*, \\ q_n \neq \underline{0_K}}} +\sum_{\substack{\text{\#}_k: \\ k\in[1:K]}}\right)\max_{j\in[1:K]}\Pb_{Q^{[j]}_n}(q_n) \right]\notag\\
    & = \log_2 \Bigg\{ \left[(N-1)p'_0 + \frac{1-Np'_0}{N^K}\right]\notag\\
    &\hspace{10em}+ (N^K-1)\frac{1-Np'_0}{N^K} + Kp'_0 \Bigg\}\notag\\ 
    & = \log_2[1+(K-1)p'_0] \notag \\
    & = \min\left(\rho,\log_2[1+(K-1)/N]\right),
\end{align}
which is exactly that given in Theorem \ref{thm:achievableML}.

\subsection{Optimizing for the Mutual Information Leakage}

The optimization problem is very similar to that in the maximal leakage case except that the objective function is the mutual information leakage. We first analytically solve for the optimal probability distribution with the new retrieval pattern excluded, i.e. $p_d = Np_0$, and then prove Theorem \ref{thm:achievableMI} using the fact that the $D_{\text{MI}}(\rho)$ is convex\cite{Lin2021}. The optimization problem is thus formulated as
\begin{align}
    \text{Minimize:}~~ & I(Q_n;M),\\
    \text{Variables:}~~ & (p_1,p_2,\dots,p_{K-1}),\\
    \text{Subject to:}~~ & -p_w \leq 0, \forall w\in 1:K-1,\\
    & \begin{aligned}
    N\sum_{w=0}^{K-1}&\binom{K-1}{w}(N-1)^w p_w -1 = 0.\label{eq:probsum}
    \end{aligned}
\end{align}
The download cost $D$ here is directly related to $p_d = Np_0$, and setting a positive $p_0$ value is equivalent to specifying $D$.

The following proposition shows that the vector $\pv$ gives an optimal solution, when the new retrieval pattern is not used.
\begin{proposition}\label{lem:mioptimal}
For each $x_{K-1}\in [1,\infty)$, the vector $\pv$ given by (\ref{eq:recursive2})-(\ref{eq:pvdef2}) is optimal for the optimization problem given above.
\end{proposition}

\begin{proof}
We first write the Lagrangian of the problem, 
\begin{align}
    \Ls = \hat I(\pv&) - \sum_{w=1}^{K-1}\lambda_w p_w \notag\\
    &+ \nu \left[\sum_{w=0}^{K-1}\binom{K-1}{w}(N-1)^w N p_w -1 \right],
\end{align}
where $\hat I(\pv)$ is defined as 
\begin{align}
        &\hat I(\pv)  = \frac{1}{K}\sum_{w=0}^{K}\cmb{K}{w}(N-1)^w \notag\\
    &\Big\{ wp_{w-1} \log p_{w-1} + (K-w)p_w \log p_w \notag\\
    & - [wp_{w-1} + (K-w)p_w] \log \frac{wp_{w-1} + (K-w)p_w}{K}\Big\}.
\end{align}
The KKT condition can be explicitly derived as follows: the partial derivatives of $\Ls$ w.r.t $p_w$ are,
\begin{align}
    &\begin{aligned}
        &\pdv{\Ls}{p_{w}} = \cmb{K-1}{w} (N-1)^w \big[- y_w - (N-1) y_{w+1}\\
        & + (N-1)\log x_{w+1}  + N \nu\big]-\lambda_w,~~w\in 1:K-2,
    \end{aligned}\\
    &\begin{aligned}
        &\pdv{\Ls}{p_{K-1}} =(N-1)^{K-1} \big[- y_{K-1} + N \nu\big]-\lambda_{K-1},
    \end{aligned}
\end{align}
where we have introduced the two new sets of variables:
\begin{align}
    x_w &\triangleq p_{w-1}/p_{w},\\
    y_w &\triangleq \log \frac{w x_w + K- w}{K}.
\end{align}
It is straightforward to verify that $x_w$ and $y_w$ satisfying \eqref{eq:recursive2}, with $y_w$'s properly eliminated, which along with the following dual variable assignments
\begin{align}
   \lambda_w &= 0,~w\in 1:K-1,\label{eq:lambda} \\
   \nu &= y_{K-1}/N,\label{eq:nu}
\end{align}
render the partial derivatives zeros for all $w\in1:K-1$, and moreover satisfy all complementary slackness requirement. These assignments are thus 
a solution to the primal optimization problem. The mutual information leakage $I(Q_n;M)$ and the download cost $D$ with this solution is exactly as the right hand sides of \eqref{eq:rhopara} and \eqref{eq:pdpara}.
\end{proof}

It is straightforward to show that the new strategy (clean download from any single server) gives the extreme point $(\rho,D) = (\frac{\log K}{N},1)$. Together with the convexity of $D_{\text{MI}}(\rho)$ (see \cite{Lin2021}), Theorem \ref{thm:achievableMI} is now obvious. In fact, for any optimized TSC code with $\pv = \tilde \pv$, probabilistic sharing with the new strategy results in $(\rho,D)$ operating points on the straight line connecting the them, and the resultant code has a positive $p'_0$ and $\pv = (1-Np'_0)\tilde \pv$. 

In Fig. \ref{fig:rhoD}, the tangent point gives the threshold beyond which the new download pattern becomes effective in the sharing solution. It can be shown after some algebra that this occurs at  $x_1 = (K-1)/(K^\frac{N-2}{N-1}-1)$. 

\section{Conclusion}
\label{sec:conclusion}
We studied the the problem of weakly private information retrieval, and proposed a new code construction based on a simple yet critical observation on the minimum download extreme case. The optimizing query pattern probability distributions are provided for the maximal leakage metric and the mutual information leakage metrics, resulting in strict improvements in both case. The new inner bounds do not yet match the known outer bounds in the literature, and we are currently working on reducing this gap.

\bibliographystyle{IEEEtran}


\end{document}